\documentclass[sigconf]{acmart}

\usepackage{graphicx} 
\usepackage{amsthm}
\usepackage{amsmath}   
\usepackage{amsbsy}    
\usepackage{subfigure}
\usepackage{subcaption}
\usepackage{graphicx}
\usepackage{tabularx} 
\usepackage{booktabs}
\usepackage{multirow}
\usepackage{multicol}
\usepackage{breqn}
\usepackage{extarrows}
\usepackage{enumitem}
\usepackage{bm}
\usepackage{xcolor}
\usepackage{pifont}
\usepackage{changepage}

\theoremstyle{plain}
\newtheorem{theorem}{Theorem}[section]

\newtheorem{lemma}[theorem]{Lemma}

\theoremstyle{definition}

\theoremstyle{remark}

\newtheorem{cond}{Condition}

\newcommand{\lc}{\left(} 
\newcommand{\rc}{\right)} 
\newcommand{\ox}{\overline{X}} 
\newcommand{\oy}{\overline{Y}} 
\newcommand{\ow}{\overline{W}} 
\newcommand{\oz}{\overline{Z}} 
\newcommand{\ty}{\widetilde{Y}}

\newcommand{\wt}{\widehat{\theta}} 
\newcommand{\wl}{\widehat{\lambda}} 
\newcommand{\wa}{\widehat{\tau}} 
\newcommand{\wv}{\widehat{\text{var}}} 
\newcommand{\wc}{\widehat{\text{cov}}} 
\newcommand{\tv}{\text{var}}
\newcommand{\tc}{\text{cov}}
\newcommand{\wdd}{\widehat{\delta}}

\newcommand{\lm}{\lim_{n\rightarrow \infty}} 
\newcommand{\lnm}{\lim_{m\rightarrow \infty}} 
\newcommand{\sm}{\sum_{i=1}^n}
\newcommand{\sn}{\sum_{i=1}^N}
\newcommand{\smm}{\sum_{i=1}^m}
\newcommand{\la}{\langle \mu_0 \rangle}
\newcommand{\lb}{\langle \mu_1 \rangle}
\AtBeginDocument{%
  }

\setcopyright{none} % 版权设置
\renewcommand\footnotetextcopyrightpermission[1]{} % removes footnote with conference information in first column

\acmConference[arXiv]{arXiv}{2026}{preprint}
\acmISBN{}

\begin{document}

\title{Ensuring Trustworthy Online A/B Testing: Addressing Five Key Questions on CUPED}

 
\author{Yu Zhang}
\authornote{Corresponding authors.} % 
\affiliation{%
  \institution{Zhongtai Securities Institute for Financial Studies, Shandong University}
  \city{Jinan}
  \country{China}}
\email{yuzhang.johnny@mail.sdu.edu.cn}

\author{Bokui Wan}
\authornotemark[1] % 
\affiliation{%
  \institution{ByteDance}
  \city{Beijing}
  \country{China}}
\email{wanbokui@bytedance.com}

\author{Yongli Qin}
\authornotemark[1] % 
\affiliation{%
  \institution{ByteDance}
  \city{Beijing}
  \country{China}}
\email{yongli.qin@bytedance.com}

\author{Jinyong Ma}
\affiliation{%
\institution{ByteDance}
\city{Beijing}
\country{China}}
\email{jinyongma@bytedance.com}

\author{Yifan Guo}
\affiliation{%
\institution{ByteDance}
\city{Beijing}
\country{China}}
\email{guoyifan.tt@bytedance.com}

\renewcommand{\shortauthors}{Zhang et al.}

\begin{abstract}

A/B testing has become the gold standard for data-driven decision-making in large-scale online experimentation, providing critical guidance for feature launch, pricing optimization, and user experience enhancement. To maximize statistical sensitivity, many technology companies routinely employ Controlled-experiment Using Pre-Experiment Data (CUPED), a technique that achieves substantial variance reduction while preserving the unbiasedness of estimating the average treatment effect. Despite its widespread adoption, several critical methodological and practical nuances of CUPED remain underexplored. 
This paper systematically addresses five frequently encountered yet overlooked questions regarding the application of CUPED. First, we provide a comparative analysis of various post-CUPED estimators to identify the optimal adjustment specification. Second, we evaluate the validity of regression-based adjustments and delineate robust variance estimation methods tailored for such frameworks. Finally, we extend our investigation to complex but common scenarios, including multi-arm experiments and two-stage sampling designs. Our findings reveal that in these settings, naive reliance on standard variance estimators can lead to severely misleading inferences. By offering rigorous theoretical insights and extensive experimental validation, this work deepens the conceptual understanding of CUPED. Notably, the recommended methodologies have been successfully deployed and integrated into ByteDance's experimentation platform. 

\end{abstract}





\maketitle

\section{Introduction}

Online controlled experiments, or A/B testing, have become the empirical cornerstone for causal inference in the digital economy \citep{deng2013improving, xie2016improving}. Leading technology enterprises, such as Google, Microsoft, Meta, and ByteDance, rely on these frameworks to conduct tens of thousands of experiments daily, guiding critical trajectories for feature deployment, pricing optimization, and user experience refinement \citep{micro, meta2024, kohavi2020trustworthy}. When executed with rigor, these experiments provide the gold standard for isolating causal effects from environmental noise, ensuring that strategic decisions are rooted in robust data rather than mere intuition \citep{box2005statistics}. 

However, the efficacy of A/B testing is frequently hampered by the inherent volatility of online user behavior \citep{freedman2008regression, deng2013improving}. The standard difference-in-means estimator \citep{splawa1990application, welch1947generalization}, while theoretically unbiased, often suffers from prohibitive sampling variance. In practice, this leads to a ``needle in a haystack'' problem: subtle but commercially significant treatment effects are often obscured by noise, resulting in insufficient statistical power and prolonged experimental durations that delay product iteration and increase opportunity costs \citep{deng2013improving, xie2016improving}. 

To mitigate these challenges, variance reduction techniques have emerged as a vital area of research \citep{freedman2008regression, lin2013agnostic, deng2013improving, chernozhukov2018double, kohavi2023online}. Among these, the Controlled-experiment Using Pre-Experiment Data (CUPED) \citep{deng2013improving} method has become the industry benchmark. By leveraging the correlation between pre-experimental covariates and post-experimental outcomes, CUPED constructs a more precise estimator without sacrificing unbiasedness \citep{deng2023augmentation}. While its adoption is widespread, as documented by Netflix \citep{xie2016improving}, Microsoft \citep{micro}, and Meta \citep{meta2024}, its implementation is far from monolithic. A critical, yet often overlooked, tension exists between model-free approaches and regression-based frameworks \citep{deng2023augmentation}. In model-free applications, practitioners frequently resort to ``plug-in'' variance estimators that substitute sample statistics for true parameters \citep{xie2016improving}. Conversely, in the regression context, the foundational work of Freedman cautioned that variance estimates based on ordinary least square (OLS) can be asymptotically biased and lead to invalid inferences under certain conditions \citep{freedman2008regression}. While Lin \citep{lin2013agnostic} later demonstrated that utilizing Huber-White ``sandwich'' estimators \citep{eicker1967limit, huber1967behavior} within an interaction-model framework can resolve these inconsistencies, the practical choice between these disparate methodologies remains a source of confusion for many analysts. 

The complexity intensifies in multi-arm experimental designs, which are increasingly prevalent in modern platforms but remain under-researched compared to the classic two-sample case \citep{freedman2008regression2, lopez2017estimation, feng2012generalized}. In scenarios involving multiple treatment groups, a fundamental question arises regarding the optimal scope of covariate adjustment: should the global covariate mean be estimated using the entire sample, or should adjustments be localized to specific comparison pairs? Naive extensions from two-arm theory to multi-arm practice are not guaranteed to hold. Indeed, improper specification of the adjustment mechanism can lead to a paradoxical loss of efficiency or, more critically, compromise the type I error rate, rendering experimental conclusions untrustworthy. 

Furthermore, two-stage sampling represents another complex scenario frequently encountered in large-scale online experiments \citep{chauvet2020inference, ohlsson1989asymptotic}. In this process, a platform first selects a specific subset of its users to be eligible for an experiment and subsequently performs random assignment only among those selected individuals. This design is primarily driven by the need to mitigate operational risk, restricting high-risk interventions to a small fraction of the population, and to comply with regulatory or legal constraints that limit user exposure to certain conditions. This two-stage approach significantly complicates the distribution of the ATE estimator by introducing compound stochasticity: the randomness of the initial selection phase and the randomness of the subsequent treatment assignment. These layered sources of variation can lead to unreliable inference and an inflation of Type I errors if the variance estimation fails to account for the underlying sampling mechanism \citep{ding2017bridging, ding2024first}. 

This paper responds to these methodological ambiguities by providing a systematic investigation into the deployment of CUPED across diverse and complex settings. Our main contributions are as follows: 

\begin{itemize}
\item We synthesize current CUPED practices across the technology companies, highlighting the critical divergence between industry heuristics and theoretical ideals. 

\item We delineate the conditions under which model-free or regression-based adjustments should be preferred, specifically emphasizing the necessity of sandwich estimators to maintain the validity of inference. 

\item We conduct a comparative analysis of estimation forms, demonstrating that group-specific adjustments yield maximum efficiency for calculating group means and relative lift. 

\item We provide rigorous guidelines for ATE inference in complex scenarios: 
  \begin{itemize}
     \item In multi-arm experiments, we prove that utilizing the full-sample covariate mean in the adjustment term is essential for optimal precision. 
     \item In two-stage sampling experiments, we derive a corrected split-sample variance estimator that accounts for compounded randomness, ensuring reliable inference results. 
   \end{itemize}
\end{itemize}

\section{Preliminaries}
\label{sec2}

\subsection{Design-Based Framework}

This study adopts a design-based framework under a completely randomized design. For each unit $i$ in a population of size $n$, let $Y_i(0)$ and $Y_i(1)$ denote the potential outcomes under control and treatment, respectively. The treatment assignment is indicated by a binary variable $T_i$, where $T_i = 1$ if unit $i$ is assigned to the treatment group and $T_i = 0$ otherwise. 

In the design-based approach, the potential outcomes $\{Y_i(0), Y_i(1)\}_{i=1}^n$ are treated as fixed, non-random constants. Randomness enters the system solely through the assignment mechanism $T_i$. The observed outcome for each unit is defined by 
$$Y_i = T_iY_i(1) + (1 - T_i)Y_i(0), \quad i=1, \ldots, n. $$
The finite-population means and the average treatment effect (ATE) are defined as 
$$\mu_1 = n^{-1} \sum_{i=1}^n Y_i(1), \quad \mu_0 = n^{-1} \sum_{i=1}^n Y_i(0), \quad \delta = \mu_1 - \mu_0.$$

We assume that each unit $i$ has an observable covariate $X_i$, unaffected by treatment. Since both potential outcomes cannot be observed simultaneously for any unit, we rely on sample-level estimates. The sample means for the observed outcomes and covariates are 
$$\oy_1 = n_1^{-1} \sum_{i=1}^n T_i Y_i, \quad \oy_0 = n_0^{-1} \sum_{i=1}^n (1 - T_i) Y_i, \quad \oy = n^{-1} \sum_{i=1}^n Y_i, $$
$$\ox_1 = n_1^{-1}  \sum_{i=1}^n T_i X_i \quad \ox_0 = n_0^{-1} \sum_{i=1}^n (1 - T_i) X_i,  \quad \ox = n^{-1} \sum_{i=1}^n X_i, $$
where $n_1 = \sum_{i=1}^n T_i$ and $n_0 = \sum_{i=1}^n (1 - T_i)$ denote the fixed group sizes. Let $p_1 = n_1/n$ and $p_0 = n_0/n$ represent the assignment probabilities. 

To ensure asymptotic properties, we impose the following regularity conditions, which are standard in design-based literature \citep{freedman2008regression, lin2013agnostic}. In principle, there should be an extra subscript to index the sequence of populations. Like Freedman and Lin \citep{freedman2008regression, lin2013agnostic}, we omit the extra subscripts for simplification. 

\begin{cond}
There is a finite bound $L$ such that for all $n$, the fourth moments of potential outcomes and covariates are bounded: 
$$
n^{-1} \sum_{i=1}^nY_i(1)^4 < L, \quad n^{-1} \sum_{i=1}^nY_i(0)^4 < L, \quad n^{-1} \sum_{i=1}^nX_i^4 < L. 
$$
 \end{cond}
 
\begin{cond}
The population means of $Y_i(1)$, $Y_i(0)$, $Y_i(1)^2$, $Y_i(0)^2$, $Y_i(1)Y_i(0)$, $Y_i(1)X_i$, $Y_i(0)X_i$, and $X_i^2$ converge to finite limits. 
 \end{cond}
 
 \begin{cond}
The assignment probability $n_1/n$ converges to a limit $p_1$, with $0 < p_1 < 1$. 
 \end{cond}

Finite-population asymptotic results are statements about randomized experiments on an imaginary infinite sequence of finite populations, with increasing $n$. Assuming Conditions $1-3$, we define the limiting parameters for $k = 0, 1$: 
$$\theta_k = \lim_{n\rightarrow \infty} \frac{\sum_{i=1}^n(Y_i(k) - \mu_k)(X_i - \ox)}{\sum_{i=1}^n(X_i - \ox)^2}, S_x^2 = \lim_{n\rightarrow \infty} \frac{1}{n-1}\sum_{i=1}^n(X_i - \ox)^2, $$
$$S_k^2 = \lim_{n\rightarrow \infty} \frac{1}{n-1}\sum_{i=1}^n(Y_i(k) - \mu_k)^2, S_{\delta}^2 = \lim_{n\rightarrow \infty} \frac{1}{n-1}\sum_{i=1}^n(Y_i(1)-Y_i(0) - \delta)^2, $$
$$$$
$$S_{k, x}^2 = \lim_{n\rightarrow \infty} \frac{1}{n-1}\sum_{i=1}^n\lc Y_i(k) - \mu_k - \theta_k(X_i - \ox)\rc^2, $$
$$S_{\delta, x}^2 = \lim_{n\rightarrow \infty} \frac{1}{n-1}\sum_{i=1}^n\lc Y_i(1)-Y_i(0) - \delta - (\theta_1 - \theta_0)(X_i - \ox)\rc^2. $$

\subsection{CUPED Overview}

The core intuition of CUPED is to leverage pre-experiment data to construct an unbiased ATE estimator with higher precision than  the simple difference-in-means. Consider an adjusted estimator for the treatment group mean: 
\begin{equation}
\ty_1 = \oy_1 - \theta \ox_1 + \theta \ox, \notag
\end{equation}
where $\theta$ is a constant. This adjustment preserves unbiasedness while yielding the variance 
\begin{equation}
\label{var1}
\mathrm{var}(\ty_1) =  \mathrm{var}(\oy_1) + \theta^2 \mathrm{var}(\ox_1) - 2\theta \mathrm{cov}(\oy_1, \ox_1), \notag
\end{equation}
where $\text{cov}$ and $\text{var}$ represent population covariance and variance, respectively. 

Minimizing this quadratic in $\theta$ yields the optimal coefficient $\mathrm{cov}(\oy_1, \ox_1)/\mathrm{var}(\ox_1)$, resulting in 
$$
\tv(\ty_1) = \mathrm{var}(\oy_1)\Big(1 - \frac{\text{cov}^2(\oy_1, \ox_1)}{\text{var}(\oy_1)\text{var}(\ox_1)}\Big). 
$$
Consequently, any non-zero covariance between the outcome and covariate enables variance reduction by filtering out the explainable variability attributable to the covariate.

\subsection{CUPED Implementation Across Companies}
\label{practice}

In industry, while the principle of CUPED remains consistent with Deng et al. \citep{deng2013improving}, variations in the estimation of $\theta$ significantly impact both variance reduction and inferential validity. We denote the ATE estimator for the $j$-th method as $\widehat{\delta}_j$, and the corresponding adjustment coefficients for the treatment and control groups as $\widehat{\theta}_{j, 1}$ and $\widehat{\theta}_{j, 0}$, respectively. 

 \textbf{Full-sample Estimator}. Common in platforms like Booking \citep{jackson2018booking}, Statsig \citep{Statsig}, and Walmart \citep{Walmart}, this approach uses a single $\theta$ estimated from the entire sample:
 $$
\widehat{\theta}_{1,1}=\widehat{\theta}_{1,0}= \frac{\widehat{\text{cov}}(\overline Y, \overline X)}{\widehat{\text{var}}(\overline X)}, 
$$
where $\wc(\oy, \ox) = n^{-1}(n - 1)^{-1}\sm(Y_i - \oy)(X_i - \ox)$ and $\wv(\ox) = n^{-1}(n - 1)^{-1}\sm(X_i - \ox)^2$. 
The resulting adjusted estimator is 
$$
\wdd_1 = \oy_1 - \oy_0 - \wt_{1, 1}(\ox_1 - \ox_0). 
$$
While straightforward, this method does not guarantee variance reduction under treatment effect heterogeneity. Notably, at ByteDance, approximately $0.1\%$ of experiments exhibited anomalous variance inflation, with standard errors expanding more than threefold compared to the difference-in-means estimator.

 \textbf{Pooled-sample Estimator}. To ensure variance reduction, this approach explicitly minimizes the variance of the ATE estimator, $\tv\big( \oy_1 - \oy_0 - \theta(\ox_1-\ox_0) \big)$, yielding the optimal coefficient estimator \citep{deng2021improving, xie2016improving}: 
$$
\wt_{2, 1} = \wt_{2, 0} = \frac{\wc(\oy_1, \ox_1) + \wc(\oy_0, \ox_0)}{\wv(\ox_1) + \wv(\ox_0)}, 
$$
where $\wc(\oy_k, \ox_k)$ and $\tv(\ox_k)$ are group-specific sample covariance and variance. For $k = 0, 1$: 
$$\wv(\ox_k) = n_k^{-1}(n_k - 1)^{-1}\sm\mathbb I(T_i = k)(X_i - \ox)^2, $$
$$\wc(\oy_k, \ox_k) = n_k^{-1}(n_k - 1)^{-1}\sm\mathbb I(T_i = k)(Y_i - \oy)(X_i - \ox). $$
The ATE estiamtor $\wdd_2$ could be obtained analogously to $\wdd_1$. While $\widehat{\delta}_2$ theoretically guarantees variance no greater than the unadjusted difference-in-means, recent studies \citep{zhang2025bridging} suggest that its standard variance estimator can be overly conservative, potentially diminishing statistical power.

 \textbf{Split-sample Estimator}. This method employs group-specific coefficients to accommodate different outcome-covariate relationships \citep{meta2024, micro}: 
$$
\text{argmin}_{\theta_1, \theta_0} \tv \lc \big(\oy_1 - \theta_1(\ox_1-\ox)\big) -  \big(\oy_0 - \theta_0(\ox_0-\ox)\big) \rc. 
$$
A specific solution that minimizes not only group-level variability but also the standard variance estimator is selected as 
$$
\wt_{3, 1}=\frac{\wc(\oy_1, \ox_1)}{\wv(\ox_1)}, \quad \wt_{3, 0}=\frac{\wc(\oy_0, \ox_0)}{\wv(\ox_0)}, 
$$
then it follows: 
$$
\widehat \delta_3 = \lc\oy_1 - \wt_{3, 1}(\ox_1 - \ox)\rc - \lc\oy_0 - \wt_{3, 0}(\ox_0 - \ox)\rc. 
$$
Zhang et al. \citep{zhang2025bridging} confirm that this approach is not only asymptotically equivalent to the pooled version but also produces less conservative variance estimates, leading to superior power in practice.

 \textbf{Regression Adjustment without Interaction}. A common alternative to model-free adjustment is the use of OLS to fit the following regression model: 
$$
Y_i = \beta_0 + \delta T_i + \beta_x(X_i - \ox) + \varepsilon_i. 
$$
This approach is essentially the regression analog of full-sample estimation. While asymptotically equivalent to $\widehat{\delta}_1$ under the design-based framework, it is susceptible to the Freedman critique \citep{freedman2008regression}: under certain conditions of model misspecification, this adjustment can paradoxically increase the variance compared to the difference-in-means.

 \textbf{Regression Adjustment with Interaction}. To address this potential variance inflation, Lin \citep{lin2013agnostic} proposed incorporating a treatment-covariate interaction term into the model: 
$$
Y_i = \beta_0 + \delta T_i + \beta_x(X_i - \ox) + \beta_IT_i(X_i - \ox) + \varepsilon_i. 
$$
The resulting ATE estimator effectively performs group-specific adjustments. Lin proved that this specification is asymptotically superior (or at least equal) to the unadjusted difference-in-means in terms of variance, mirroring the advantages seen in the split-sample estimation.

\section{Regression-Based or Model-Free Approaches}
\label{sec3}

As detailed in the preceding sections, technology companies employ a diverse array of strategies for implementing CUPED. The first three techniques operate without relying on any underlying models or distributional requirements, they derive adjustments directly from observed sample statistics, differing primarily in how they aggregate the covariance structure between outcomes and covariates. By contrast, the remaining two approaches utilize linear regression, modeling the relationship between the outcome and covariates through a linear functional form and estimating parameters via OLS. 

While traditional linear regression theory often invokes stringent requirements, such as independently and identically distributed errors, it is a common misconception that model-free methods are inherently more robust in experimental settings. Recent research \citep{zhang2025bridging} demonstrates that for mean-based metrics, estimators from both categories exhibit asymptotic equivalence. Specifically, the pooled-sample, split-sample, and interaction-inclusive regression estimators asymptotically converge to the same limiting distribution. Similarly, the full-sample and non-interaction regression estimators are asymptotically identical (see Figure 1 in \citep{zhang2025bridging}). Crucially, this equivalence holds even under model misspecification, ensuring that regression-based results remain as dependable as their model-free counterparts within the design-based framework. 

For more complex measures, such as ratio metrics like click-through rates, regression modeling can prove cumbersome to implement and interpret \citep{deng2023augmentation}. In these instances, model-free techniques are superior due to their flexibility, allowing seamless integration with the delta method to compute adjustments and test statistics. Conversely, regression offers significant advantages when incorporating high-dimensional covariates, where established OLS software packages streamline the adjustment process far more efficiently than manual model-free calculations. Ultimately, the choice between these paradigms should not be dictated by concerns over violated model assumptions, which are mitigated by asymptotic properties; rather, the decision should hinge on computational efficiency and the mathematical structure of the metric under study.

\section{Construction of the Adjusted Outcome Mean}
\label{sec4}

Beyond the estimation of the ATE, practitioners are often tasked with reporting group means and relative lift. While different constructions of adjusted means may yield the same point estimate for the ATE, their performance in estimating group-level parameters and relative lift can diverge significantly. 

This section evaluates three prevalent methods for constructing adjusted group means, denoted by $\widetilde{Y}$. To ensure a robust comparison, we utilize adjustment coefficients from the split-sample method, as they align asymptotically with Lin's \citep{lin2013agnostic} interaction-inclusive regression and maximize testing power.

The three primary constructions observed in industry \citep{Statsig, Statsig2, Nubank} are: 
\begin{equation}\label{2}
\ty_1 = \oy_1 - \wt_{3, 1} \ox_1, \quad \ty_0 = \oy_0 - \wt_{3, 0} \ox_0. \notag
\end{equation}
\begin{equation}\label{3}
\ty_1^* = \oy_1 - \wt_{3, 1} ( \ox_1 - \ox ), \quad \ty_0^* = \oy_0 - \wt_{3, 0} (\ox_0 - \ox). \notag
\end{equation}
\begin{equation}\label{4}
\ty_1^{**} = \oy_1 - \wt_{3, 1} (\ox_1 - \ox) + \wt_{3, 0} (\ox_0 - \ox), \quad \ty_0^{**} = \oy_0. 
\end{equation}

The first construction, while simple, shifts the expectation of the metric away from the true population mean, making it unsuitable for calculating relative lift. Consequently, our analysis focuses on the comparison between the second and third constructions. 

The third construction is particularly popular in industry because it leaves the control group mean unchanged. For platforms serving users with varying levels of statistical expertise, this ``control-preserving'' property is a major psychological advantage: it eliminates the need to explain why the ``baseline'' number has moved after applying CUPED. 

While this approach preserves the ATE point estimate and its associated $p$-value, it incurs a significant efficiency cost. By shifting the entire adjustment term onto the treatment group, the third construction introduces unnecessary noise into the group-level and relative lift estimates. The following theorem precisely captures these distinctions. 

\begin{theorem}\label{relative}
Under the design-based framework, let $\tv(\ty^*_1/\ty^*_0)$ and $\tv(\ty^{**}_1/\ty^{**}_0)$ denote the true variances of the relative lift estimators for the second and third constructions, respectively. Let $\lb = \lim\limits_{n\rightarrow \infty} \mu_1$ and $\la = \lim\limits_{n\rightarrow \infty} \mu_0$. Then, the variance difference $n\big( \tv(\ty^{**}_1/\ty^{**}_0) - \tv(\ty^{*}_1/\ty^{*}_0)\big) $ converges to 
$$ \frac{\lb^2}{\la^2} \lc \frac{1}{\lb} - \frac{1}{\la} \rc^2 \frac{p_1}{p_0}\theta_0^2S_x^2 \ge 0.  $$ 
Similarly, $n\big(\tv(\ty^{**}_1) - \tv(\ty^{*}_1)\big) $ converges to 
$p_0^{-1}p_1\theta_0^2S_x^2$, and \\
$n\big(\tv(\ty^{**}_0) - \tv(\ty^{*}_0)\big) $ converges to 
$p_0^{-1}p_1\theta_0^2S_x^2$. 
\end{theorem}

\begin{figure}[!htb]
  \centering 
  \subfigure[Empirical distributions of the group-mean estimators in the treatment group. ]
  {
   \includegraphics[width=0.225\textwidth]{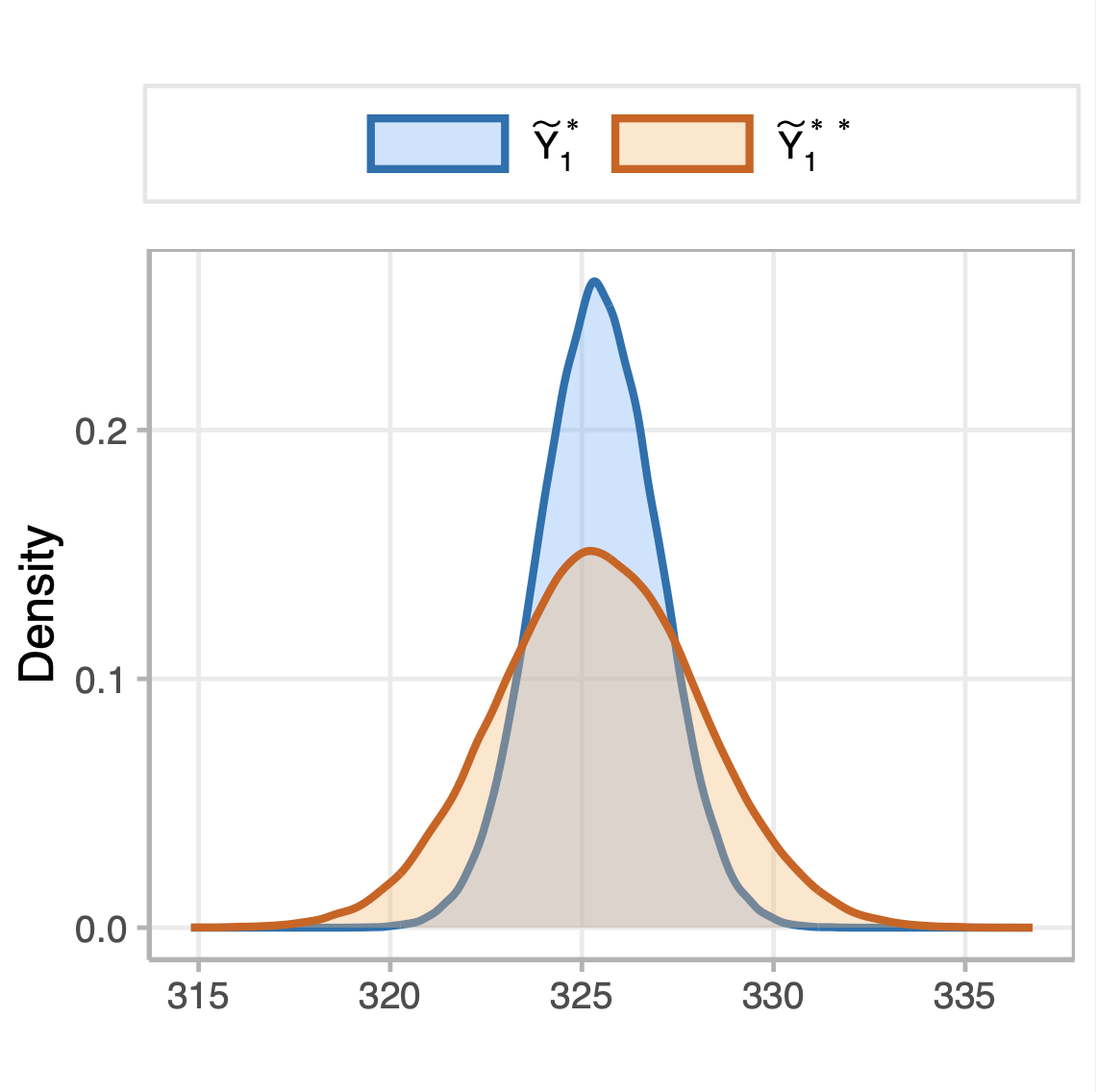}
  }
    \subfigure[Empirical distributions of the relative lift estimators. ]
  {
   \includegraphics[width=0.225\textwidth]{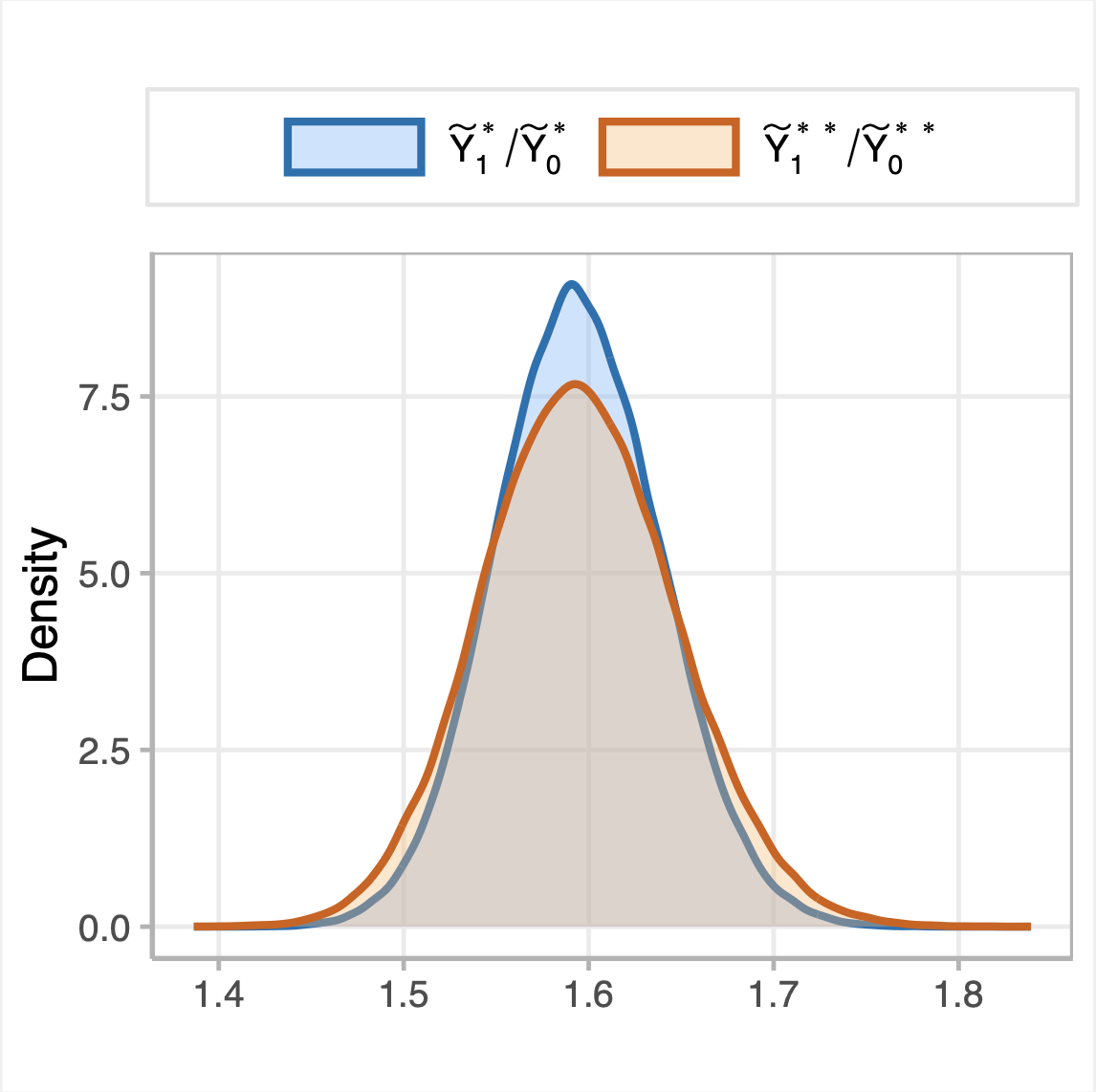}
  }

  \caption{Empirical distributions of (a) group-mean estimators in the treatment group and (b) relative lift estimators based on $B=10^5$ replications. The distribution of $\widetilde{Y}_1^{*}/\widetilde{Y}_0^{*}$ is tighter than $\widetilde{Y}_1^{**}/\widetilde{Y}_0^{**}$, demonstrating the superior estimation of the former. }
  \label{rela}
\end{figure}

Beyond this formal insight, an empirical illustration using real-world data intuitively highlights the benefits. Specifically, we repeatedly partition the dataset into treatment and control groups with equal probability, then compute the estimators of the group means and the relative lift over $B=10^5$ replications to obtain their empirical distributions. As shown in Figure \ref{rela}, the empirical distribution for estimates from the second construction is markedly tighter. While both yield unbiased results, the second delivers superior precision in real applications. 

Finally, when using regression-based methods to estimate group means, one can first obtain the estimated adjustment coefficients via OLS, then calculate the adjusted outcomes at the unit level, and use these to construct the group mean estimators and the relative lift estimator.

\section{Variance Estimation and Inferential Validity}
\label{sec5}

In model-free adjustment frameworks, variance estimation is relatively straightforward for mean-based metrics. Because adjustments are applied directly to individual observations, the variance of the ATE can be estimated using the standard sample variance of these transformed values. For ratio metrics, the delta method provides a reliable approximation \citep{ver2012invented}. Within the design-based inference framework, these estimators have been proven to be either consistent or asymptotically conservative, ensuring that Type I error rates remain controlled \citep{ding2024first}. 

However, when employing regression models to estimate treatment effects, a critical risk arises. Standard statistical software typically reports variance based on OLS assumptions (e.g., homoscedasticity), which can be highly unreliable under a design-based framework. Specifically, OLS variance estimators may overestimate or underestimate the true variance depending on the underlying data distribution \citep{freedman2008regression}. This challenge is mitigated by adopting robust variance estimation, commonly known as the sandwich estimator (or Huber-White estimator). The sandwich estimator is designed to remain valid under heteroscedastic patterns and does not require strong parametric assumptions regarding the residual variance \citep{eicker1967limit, huber1967behavior}. In the context of the interaction-inclusive regression model, the sandwich estimator is guaranteed to be asymptotically consistent or conservative, providing a theoretically sound basis for inference \citep{lin2013agnostic}. 

The discrepancy between OLS and sandwich estimators becomes particularly pronounced when residual variances diverge across groups and group sizes are imbalanced \citep{white1980heteroskedasticity, kauermann2000sandwich}. To demonstrate this, we conducted an experimental validation using a real-world dataset. Following the logic in Section \ref{sec4}, we performed $B = 10^4$ replications under the null hypothesis, manually introducing heteroscedasticity between treatment and control groups. Table \ref{varest} displays the Type I error rates across varying treatment assignment probabilities. 

\begin{comment}
To demonstrate this intuitively, we conducted a revealing simulation experiment. We simulated data for $10^4$ units, fixing potential outcomes under a defined generation rule and repeating random group assignments $10^4$ times. We explored balanced and imbalanced allocation scenarios, with treatment probabilities of $0.25$, $0.50$, and $0.75$. The data generation process follows 
$$
\begin{cases}
Y(0) = 1 + (X - \ox) + 0.2(X^2 - \overline{X^2}) + (\varepsilon_1 - \overline{\varepsilon}_1), \\
Y(1) = 1 + (X - \ox) + 0.1(\varepsilon_0 - \overline{\varepsilon}_0) + 0.9(\varepsilon_1 - \overline{\varepsilon}_1),  
\end{cases}
$$
where $\varepsilon_1 \overset{i.i.d}{\sim}\varepsilon_0\sim N(0, 1)$, $X\sim N(0, 4)$, and $\overline{A}$ represents the mean of variable $A$. 

Since the data were generated with zero ATE, the null hypothesis is exactly true in the population. Consequently, the empirical Type I error rates directly reflect the validity of the inference procedure. The resulting Type I error rates for using OLS and sandwich variance estimators are reported in Table \ref{varest}. 
\end{comment}

\begin{table}[!htb]
    \centering
   \caption{Comparison of Type I error rates between the sandwich and OLS variance estimators under the null hypothesis. Results are based on $B = 10^4$ replications with varying treatment assignment probabilities to assess performance under balanced and imbalanced allocations. }
    \resizebox{0.48\textwidth}{!}{
    \begin{tabular}{l|ccc}
    \toprule
    Probability of receiving treatment & $0.25$ & $0.50$ & $0.75$ \\
    \midrule
    Sandwich estimator & $4.20\%$ & $3.94\%$ & $4.64\%$ \\
    \midrule
    OLS variance estimator & $2.30\%$ & $3.93\%$ & $6.82\%$ \\
    \bottomrule
    \end{tabular}}
    \label{varest}
\end{table}

The results indicate that the sandwich estimator maintains stable, conservative control over the Type I error rate regardless of group allocation. In contrast, the OLS estimator is only reliable in balanced setups. In imbalanced scenarios, it either becomes excessively conservative, eroding detection power, or dangerously liberal, inflating false positives and leading to overconfident conclusions. This underscores why the sandwich estimator is the preferred tool for regression-based adjustment in industrial A/B testing. 

Finally, for simple mean metrics, we highlight a hybrid approach: practitioners may use OLS to obtain the ATE point estimate but calculate the variance by taking the sample variance of the CUPED-adjusted observations. As proved by Zhang et al. \citep{zhang2025bridging}, this hybrid method is asymptotically equivalent to the sandwich estimator derived from an interaction-inclusive regression model. Consequently, both paths offer statistically indistinguishable and robust inference in large samples.

\section{Multi-Arm Experiments}
\label{sec6}

In modern large-scale online experimentation, experimenters frequently assess several new strategies at once, with over half of all trials at ByteDance involving two or more treatment groups. This common practice raises a critical and practically meaningful question: when applying split-sample CUPED, should the pooled covariate mean be computed globally using all units in the experiment, or locally using only the two groups currently being compared? To maintain analytical clarity, we focus on a three-arm setup: a control group and two treatment groups. This framework captures the essential complexities of multi-arm inference, as any further increase in treatment arms follows an analogous logic. Throughout this section, we utilize the split-sample estimation approach; its group-specific adjustment coefficients ensure that the estimation for one group remains robustly shielded from potential data anomalies or treatment effect heterogeneity in another.

For each unit $i=1, \ldots, m$, let $A_i\in \{0, 1, 2\}$ denote the assignment. The notations $W_i(0)$, $W_i(1)$, and $W_i(2)$ represent the potential outcomes under control, treatment 1, and treatment 2, respectively. The observed response $W_i$ is expressed via the equation: 
$$
W_i=\sum_{k=0}^2\mathbb I(A_i=k)W_i(k), \quad i = 1, \ldots, m, 
$$
where $\mathbb I(\cdot)$ is the indicator function. Each unit is associated with a pre-treatment covariate $Z_i$, which is unaffected by the treatment assignment. Denote the ATE and group-specific means of the outcomes and covariates as 
$$
\tau = m^{-1} \sum_{i=1}^m \lc W_i(1) - W_i(0)\rc, \quad \oz = m^{-1} \sum_{i=1}^m Z_i, 
$$
$$
\ow_k = m_k^{-1} \sum_{i=1}^m \mathbb I(A_i=k) W_i, \quad \oz_k = m_k^{-1} \sum_{i=1}^m \mathbb I(A_i=k) Z_i, 
$$
and the pooled mean for the control and first treatment group by
$$
\oz_{0, 1} = (m_0+m_1)^{-1} \sum_{i=1}^m \big(\mathbb I(A_i=0) + \mathbb I(A_i=1)\big)Z_i, 
$$
where $m_k= \sum_{i=1}^m \mathbb I(A_i=k), k=0, 1, 2$. Furthermore, denote $q_k=m_k/m$. 

We now compare the performance of the split-sample estimation that uses the covariate mean computed over all groups versus the version that uses the covariate mean computed only over the two groups being directly compared. Mathematically, for $k = 0, 1$, let 
$$
\widetilde{W}_k^{*} = \ow_k - \widehat{\lambda}_k(\oz_k - \oz), \quad  \widehat{\tau}_1 = \widetilde{W}_1^{*} - \widetilde{W}_0^{*}, 
$$
$$
\widetilde{W}_k^{**} =  \ow_k - \widehat{\lambda}_k(\oz_k - \oz_{0, 1}), \quad  \widehat{\tau}_2 = \widetilde{W}_1^{**} - \widetilde{W}_0^{**}, 
$$
where $\widehat{\lambda}_k = \wc(\ow_k, \oz_k )/\wv(\oz_k)$ is the group-specific adjustment coefficient. 

Intuition suggests that $\widehat{\tau}_1$ should be more efficient because $\oz$ incorporates information from more units, providing a more stable estimate of the population covariate mean. Our theoretical results confirm this, inference based on $\widehat{\tau}_1$ have greater statistical power, and individual treatment effect estimates are also more precise. 
\begin{theorem}\label{61}
Let $\tv(\widehat{\tau}_1)$ and $\tv(\widehat{\tau}_2)$ denote the true variances of $\widehat{\tau}_1$ and $\widehat{\tau}_2$, respectively. Further, let $S_{1, z}^2, S_{0, z}^2, S_{\tau, z}^2$, and $S_z^2$ be defined analogously to $S_{1, x}^2, S_{0, x}^2, S_{\delta, x}^2$, and $S_x^2$. Then $m\tv(\widehat{\tau}_2)$ converges to 
$$
\frac{1}{q_1}S_{1, z}^2 + \frac{1}{q_0} S_{0, z}^2 - S_{\tau, z}^2 + \frac{q_2}{1-q_2}(\lambda_1 - \lambda_0)^2 S_z^2, 
$$
which is greater than or equal to the true asymptotic variance of $\sqrt{m}\widehat{\tau}_1$. The difference is 
$$
\frac{q_2}{1-q_2}(\lambda_1 - \lambda_0)^2 S_z^2. 
$$
\end{theorem}

This result reveals that unless the treatment effects on the outcome-covariate relationship are perfectly homogeneous, using the local mean $\oz_{0,1}$ leads to a loss of efficiency. This is because the local mean fails to align with the optimal variance-minimizing target of the coefficients, wasting the precision gains offered by the multi-arm design.

While the loss of efficiency is problematic, a far more severe issue concerns variance estimation. Practitioners often assume that a standard post-CUPED variance estimator, which is valid in two-arm settings, can be safely applied to $\widehat{\tau}_2$. As Theorem \ref{62} demonstrates, this assumption may lead to an untrustworthy inference result. 

\begin{theorem}\label{62}
Let $\wv(\widehat{\tau}_1)$ and $\wv(\widehat{\tau}_2)$ denote the standard post-CUPED variance estimators of $\widehat{\tau}_1$ and $\widehat{\tau}_2$, defined as
$$
\begin{aligned}
&\quad \wv(\widehat{\tau}_1) = \wv(\widehat{\tau}_2) \\
&= \frac{1}{m_1(m_1-1)}\sum_{i=1}^m\mathbb I(A_i=1)\big( W_i - \ow_1 - \widehat{\lambda}_1(Z_i - \oz_1)\big)^2 \\
&\quad + \frac{1}{m_0(m_0-1)}\sum_{i=1}^m\mathbb I(A_i=0)\big( W_i - \ow_0 - \widehat{\lambda}_0(Z_i - \oz_0)\big)^2.
\end{aligned}
$$
Then $m\wv(\widehat{\tau}_1)$ and $m\wv(\widehat{\tau}_2)$ converges in probability to 
$$q_1^{-1}S_{1, z}^2 + q_0^{-1} S_{0, z}^2, $$
which exceeds the true asymptotic variance of $\sqrt{m}\widehat{\tau}_1$ by $S_{\tau, z}^2$, but does not necessarily exceed the asymptotic variance of $\sqrt{m}\widehat{\tau}_2$. The difference in the latter case is 
$$S_{\tau, z}^2 - \frac{q_2}{1-q_2}(\lambda_1 - \lambda_0)^2 S_z^2. $$
\end{theorem}

Theorem \ref{62} highlights a critical risk: the standard variance estimator of $\widehat{\tau}_2$ may underestimate the true uncertainty, leading to inflated Type I errors and untrustworthy $p$-values. To rectify this, we propose a correction term for the local approach: 
$$\wv(\widehat{\tau}_2)_{\text{corrected}} = \wv(\widehat{\tau}_2) + \frac{1}{m}\cdot\frac{q_2}{1-q_2}(\widehat{\lambda}_1 - \widehat{\lambda}_0)^2\cdot \frac{1}{m-1}\sum_{i=1}^m(Z_i-\oz)^2. $$
This correction restores the conservative property of the variance estimator, though it cannot recover the intrinsic efficiency lost by using $\widehat{\tau}_2$. 

To more intuitively illustrate the practical impact of choosing $\oz$ versus $\oz_{0, 1}$, we evaluate the performance of $\widehat{\tau}_1$ and $\widehat{\tau}_2$ through real-world data. 
We evaluate Type I error rates and power across three allocation schemes in Table \ref{61tab}, comparing standard variance estimator of $\widehat{\tau}_1$, standard variance estimator of $\widehat{\tau}_2$, and refined variance estimator $\tv(\widehat{\tau}_2)_{\text{corrected}}$. We omit $\widehat{\tau}_2$ without the variance correction from the empirical power comparisons because of its inflated Type I error rates.

\begin{table}[!htb]
    \centering
    \caption{Type I error rates and empirical power of ATE estimators $\widehat{\tau}_1$ (using full-sample covariate mean $\oz$), $\widehat{\tau}_2$ (using two-group covariate mean $\oz_{0, 1}$ without variance correction), and $(\widehat{\tau}_2)_{\text{corrected}}$ (with variance correction) across three allocation schemes $(q_0, q_1, q_2)$. }
    \resizebox{0.47\textwidth}{!}{
    \begin{tabular}{l|ccc|cc}
    \toprule
    \multirow{2}{*}{$(q_0, \ q_1, \ q_2)$} & \multicolumn{3}{c|}{Type I error rates} & \multicolumn{2}{c}{Empirical power} \\
   \cmidrule(lr){2-4} \cmidrule(lr){5-6}
    & $\widehat{\tau}_2$ & $(\widehat{\tau}_2)_{\text{corrected}}$ & $\widehat{\tau}_1$ & $(\widehat{\tau}_2)_{\text{corrected}}$ & $\widehat{\tau}_1$ \\
    \midrule
    $(0.25, 0.25, 0.50)$ & $18.46\%$ & $4.64\%$ & $4.18\%$ & $36.26\%$ & $68.37\%$ \\
    $(0.33, 0.33, 0.33)$ & $13.72\%$ & $4.41\%$ & $3.81\%$ & $53.32\%$ & $80.40\%$ \\
    $(0.40, 0.40, 0.20)$ & $9.48\%$ & $4.19\%$ & $3.59\%$ & $70.56\%$ & $87.11\%$ \\
    \bottomrule
    \end{tabular}}
    \label{61tab}
\end{table}

As shown in Table \ref{61tab}, the uncorrected variance estimator for $\widehat{\tau}_2$ produces Type I error rates that rises with $q_2$, reflecting the increasing influence of the omitted variance component. The corrected version, by contrast, reliably maintains the nominal error rate. In terms of power, $\widehat{\tau}_1$ substantially outperforms even the corrected version of $\widehat{\tau}_2$, underscoring the efficiency gains from using the covariate mean computed over all groups. 

\begin{figure}[htbp]
  \centering 
  \includegraphics[width=0.45\textwidth]{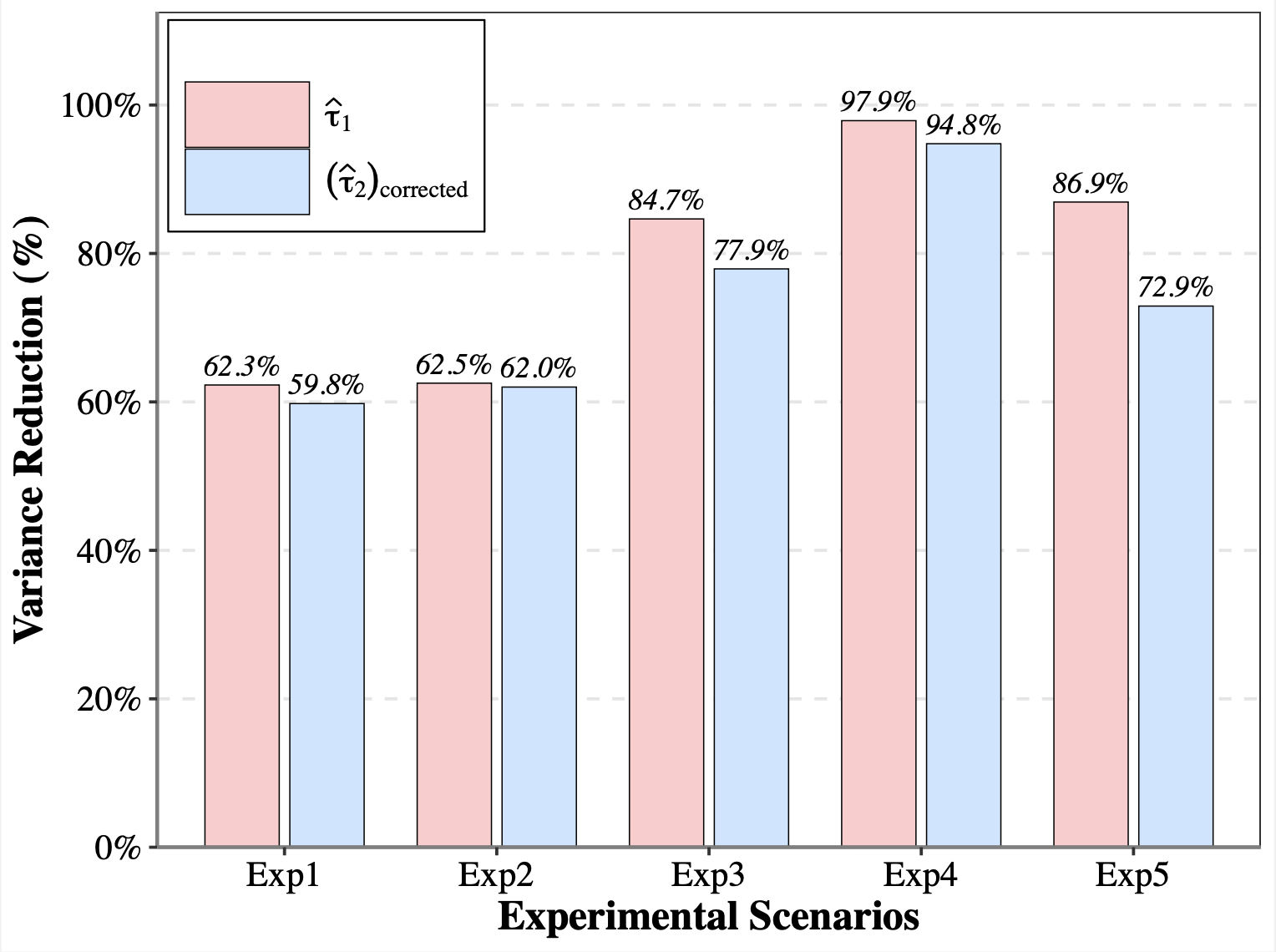}
  \caption{Variance reduction results of $\widehat{\tau}_1$ and $(\widehat{\tau}_2)_{\text{corrected}}$ relative to the standard difference-in-means estimator across five real-world experiments. }
  \label{real_vr}
\end{figure}

Finally, we utilize real-world data from ByteDance's experimentation platform to evaluate the performance of $\widehat{\tau}_1$ and $(\widehat{\tau}_2)_{\text{corrected}}$. All displayed experiments are derived from the platform's core business metrics, encompassing GMV and user feedback such as likes, dislikes, comments, and shares. To maintain data privacy, we have omitted the specific metric names and instead refer to them as Exp1 through Exp5. We must clarify that the selected experiments are intended for illustrative purposes; a significant reduction in additional variance is not achieved across all online experiments, as many of these include a large number of A/A tests. 

Figure \ref{real_vr} illustrates the variance reduction results of $\widehat{\tau}_1$ and $\widehat{\tau}_2$ with variance correction relative to the standard difference-in-means estimator across these five experiments. Even in the worst scenario, $\widehat{\tau}_1$ does not underperform relative to $(\widehat{\tau}_2)_{\text{corrected}}$ and consistently achieves a higher variance reduction ratio. In specific scenarios, such as Exp5, using the full-sample covariate mean yield an additional $10\%$ variance reduction. For a large-scale online experimentation platform, a $10\%$ gain represents a massive reduction in the traffic required to reach significance, directly lowering experimental costs and accelerating the product iteration cycle.

In summary, utilizing the full-sample covariate mean for adjustment under multi-arm settings is mathematically superior, as it obviates the need for complex variance corrections and maximizes the statistical sensitivity of the experiment simultaneously.

\section{Two-Stage Sampling}
\label{sec7}

This section addresses two-stage sampling, a sophisticated experimental design that is ubiquitous in industrial platforms. 
In large-scale digital platforms, it is rarely prudent to expose the entire user base to a single experiment. Instead, a two-stage sampling design is typically employed: first, the platform randomly selects a subset of users (e.g., selecting $100$ out of $1,000$ randomized buckets) to be eligible for the experiment; second, the selected individuals are randomized into treatment and control groups. This design serves as a critical safety valve, ensuring that if a new strategy induces a negative user experience or revenue loss, the impact is strictly confined to a small slice of the population. 

Under this framework, we treat the entire user base as the population of interest. Randomness is compounded: it arises not only from the group assignment but also from the initial selection of the experimental sample. Prior research on population-wide experiments suggests that the split-sample estimator $\wdd_3$ is generally superior to the full-sample $\wdd_1$ and pooled-sample $\wdd_2$ estimators. However, in two-stage sampling, a new nuance emerges: the covariate mean $\ox$, which is a fixed constant in population-wide tests, becomes a stochastic quantity. 
This randomness now contributes to the sampling variance of the ATE estimator. In contrast, $\wdd_1$ and $\wdd_2$ do not inherit this extra uncertainty because they rely on a shared adjustment coefficient across treatment arms. Thus, assuming that $\wdd_3$ automatically retains its full-population advantage may lead to misleading conclusions. 

A critical question is whether the randomness of the initial sampling stage erodes the efficiency gains of split-sample adjustment. Fortunately, our theoretical analysis demonstrates that $\widehat{\delta}_3$ retains its asymptotic advantage even under the added layer of selection noise. 

\begin{theorem}\label{71}
Assume the platform contains a total $N$ users. In each experiment, $n \le N$ individuals are randomly sampled for participation, with $n_1$ assigned to treatment and $n_0$ to control. Let $p_1=\lim\limits_{n \rightarrow\infty }n_1/n$, $p_0=\lim\limits_{n \rightarrow\infty }n_0/n$, and suppose $\lim\limits_{n, N \rightarrow\infty }n/N=p$. Define $\theta_{\text{Full}} = p_0\theta_0 + p_1\theta_1$ and $\theta_{\text{Pooled}} = p_1\theta_0 + p_0\theta_1$. Then as $n$ tends to infty, $n\tv(\wdd_1)$ converges to 
$$
\frac{1}{p_1}S_1^2 + \frac{1}{p_0}S_0^2 - pS_{\delta}^2 + \frac{1}{p_0p_1}\lc\theta_{\text{Full}}^2 - 2\theta_{\text{Full}}\theta_{\text{Pooled}}\rc S_x^2, 
$$
which is greater than or equal to the true asymptotic variance of $\sqrt{n}\wdd_2$ and $\sqrt{n}\wdd_3$. The difference is 
$$
\frac{1}{p_0p_1}(\theta_{\text{Full}} - \theta_{\text{Pooled}})^2S_x^2. 
$$
\end{theorem}

\begin{table*}[!ht]
    \centering
   \caption{Confidence interval coverage rates and empirical power under two-stage sampling. Experiments are conducted with fixed sampled size $n=10,000$, treated sample size $n_1=5,000$, and varying first-stage sampling probabilities $p=0.50, 0.25, 0.125$ (corresponding to platform population sizes \(N=2\times10^4, 4\times10^4, 8\times10^4\)), based on $B = 10^4$ replications. }
    \resizebox{0.85\textwidth}{!}{
    \begin{tabular}{l|cccc|ccc}
    \toprule
      \multirow{2}{*}{Total size $N$, sampling probability $p$} & \multicolumn{4}{c|}{Confidence Interval coverage rates} & \multicolumn{3}{c}{Empirical power}  \\
   \cmidrule(lr){2-5} \cmidrule(lr){6-8}
      & $\wdd_1$ & $\wdd_2$ & $\wdd_3$ & $(\wdd_3)_{\text{corrected}}$ & $\wdd_1$ & $\wdd_2$ & $(\wdd_3)_{\text{corrected}}$ \\
     \midrule
    $N=2\times 10^5, \qquad p = 0.50$ & $95.46\%$ & $95.45\%$ &  $94.82\%$ & $95.45\%$ & $83.39\%$ & $83.46\%$ &  $83.46\%$ \\
    $N=4\times 10^5, \qquad p = 0.25$ & $95.22\%$ & $95.22\%$ &  $94.62\%$ & $95.22\%$ & $82.81\%$ & $82.96\%$ & $82.96\%$ \\
    $N=8\times 10^5, \qquad p = 0.125$ & $95.02\%$ & $95.03\%$  & $94.33\%$ & $95.03\%$ & $83.04\%$ & $83.14\%$ & $83.14\%$ \\
    \bottomrule
    \end{tabular}}
    \label{ts sample}
\end{table*}

Theorem \ref{71} confirms that even when the experiment observes only a fraction of the population ($p < 1$), the split-sample approach $\wdd_3$ continues to achieve the same optimal efficiency properties documented in full-population settings. Furthermore, the pooled-sample estimator $\wdd_2$ is asymptotically equivalent to $\wdd_3$, and is at least as efficient as the full-sample estimator $\wdd_1$ and the difference-in-means. 

While efficiency is preserved, inferential validity is not. If we view the subsampled participants as the ``active'' groups and the non-sampled users as a ``hidden third group'', the two-stage setting becomes mathematically isomorphic to the multi-arm framework discussed in Section \ref{sec6}. Consequently, the standard variance estimator for $\wdd_3$ fails to account for the stochasticity of the non-participating population.

\begin{theorem}\label{72}
For $j=1, 2, 3$, denote the standard post-CUPED variance estimators as $\wv(\wdd_j)$. 
Then, $n\wv(\wdd_1)$ converges in probability to 
$$\frac{1}{p_1}S_1^2 + \frac{1}{p_0}S_0^2 + \frac{1}{p_0p_1}\lc\theta_{\text{Full}}^2 - 2\theta_{\text{Full}}\theta_{\text{Pooled}}\rc S_x^2, $$
which exceeds or equals the true asymptotic variance of $\sqrt{n}\wdd_1$ by $pS_{\delta}^2$. 
Likewise, $n\wv(\wdd_2)$ converges in probability to 
$$\frac{1}{p_1}S_1^2 + \frac{1}{p_0}S_0^2 - \frac{1}{p_0p_1}\theta_{\text{Pooled}}^2S_x^2, $$
which is greater than or equal to the true asymptotic variance of $\sqrt{n}\wdd_2$, the difference is $pS_{\delta}^2$. 
In contrast, $n\wv(\wdd_3)$ converges in probability to 
$$\frac{1}{p_1}S_1^2 + \frac{1}{p_0}S_0^2 - \lc \frac{1}{p_1}\theta_1^2 + \frac{1}{p_0}\theta_0^2\rc S_x^2, $$
which does not necessarily exceed or equal the true asymptotic variance of $\sqrt{n}\wdd_3$. The difference is 
$$pS_{\delta}^2 - (\theta_1-\theta_0)^2S_x^2, $$
revealing when the standard post-CUPED variance estimator becomes unreliable.
\end{theorem}

As the sampling probability $p$ decreases, the term $p S_{\delta}^2$ vanishes, causing the standard estimator to systematically underestimate the true variance whenever $\theta_1 \neq \theta_0$. This leads to an unacceptable inflation of Type I errors. The remedy is equally parallel. As in Theorem \ref{62}, adding an explicit correction term to the variance estimate yields a variance estimator that is either conservative or fully consistent for the true variance of $\wdd_3$. 

To illustrate this issue, we conducted an experiment using real-world data. We fixed the number of sampled individuals at $n = 10,000$ and the treated sample at $n_1=5,000$. We considered first-stage sampling probabilities $p = 0.5, 0.25, 0.125$, corresponding to the platform population sizes $N = 2 \times 10^4$, $4 \times 10^4$, and $8 \times 10^4$. 
For each setting, we performed $B=10^4$ replications and recorded both the empirical confidence-interval coverage and statistical power, summarized in Table \ref{ts sample}.

The results confirm the theoretical prediction: the uncorrected version of $\wdd_3$ systematically misestimates variance, leading to under-coverage of the nominal confidence intervals. In contrast, the corrected estimator $(\wdd_3)_{\text{corrected}}$ eliminates this bias and maintains appropriate coverage across all sampling probabilities. 

Furthermore, for estimators $\widehat{\delta}_1$, $\widehat{\delta}_2$, and $(\widehat{\delta}_3)_{\text{corrected}}$, the table reveals a broader pattern driven by the first-stage sampling probability, $p$. The variance estimators contain an intrinsic upward bias proportional to $pS_{\delta}^2$. As $p$ declines, this bias shrinks, causing the coverage rates of all estimators to converge from conservative levels to the nominal level. In the limiting case where $p$ approaches zero, this bias vanishes entirely, yielding a consistent variance estimator. Conceptually, this limit ($p \rightarrow 0$) corresponds to a scenario where the platform population is effectively infinite relative to the experiment size, naturally aligning the analysis with a super-population framework. The two-stage sampling design, however, provides a unified perspective by allowing $p$ to vary continuously from $1$ to $0$, creating a smooth bridge between the finite-population and super-population views. This makes transparent how the extent of variance overestimation depends directly and continuously on the sampling probability. 

In terms of power, the results follow naturally from the variance comparisons. Since $\wdd_2$ and $(\wdd_3)_{\text{corrected}}$ share the same asymptotic variance and use variance estimators that are asymptotically equivalent, both achieve virtually identical test power. Moreover, both consistently outperform $\wdd_1$, a direct consequence of their strictly smaller asymptotic variance. 

In summary, for two-stage sampling, practitioners should adopt the pooled-sample estimator or the corrected split-sample estimator. Both approaches retain the high statistical power of CUPED while ensuring that the corresponding confidence intervals remain valid even as the experiment scales relative to the platform’s total population.

\section{Conclusions}
\label{sec8}

This paper has provided a systematic investigation into five pivotal methodological and practical challenges often overlooked in the deployment of CUPED. As variance reduction becomes an essential requirement for high-velocity online experimentation, understanding the nuances of estimator selection and inferential validity is paramount. Through a rigorous integration of finite-population asymptotic theory and extensive empirical validation, we have delineated the boundary conditions for various CUPED implementations. 

Our findings yield several high-impact recommendations for experimental practitioners: 
\begin{itemize} 
\item \textbf{Methodological Selection:} While regression-based CUPED is superior for handling high-dimensional covariates due to its computational convenience, model-free variant remains the gold standard for ratio metrics, where it offers greater flexibility through the delta method. 
\item \textbf{Optimal Adjustment:} In constructing adjusted group means, we proved that the group-specific adjustment consistently outperforms other alternatives in terms of statistical efficiency, while providing narrower uncertainty intervals for relative lift. 
\item \textbf{Inferential Robustness:} For regression-based frameworks, the Huber-White sandwich estimator is indispensable. It safeguards against the risks of heteroscedasticity and group imbalance, where standard OLS variance estimators fail to provide valid coverage. 
\item \textbf{Multi-Arm Strategies:} In experiments with multiple treatment arms, utilizing the  full-sample covariate mean is strictly more efficient than pairwise adjustments. This approach simplifies the implementation while maximizing the use of available information. 
\item \textbf{Two-Stage Sampling:} We demonstrated that the split-sample estimator maintains its efficiency advantages even under compound stochasticity. However, to ensure valid inference in two-stage sampling designs, a variance correction term must be applied to account for the randomness of the initial sampling phase. 
\end{itemize}

These insights offer a robust and actionable framework for employing CUPED across increasingly complex testing environments. By standardizing the choice of adjustment strategies and variance estimators, organizations can significantly enhance their experimental sensitivity while eliminating the risk of drawing misleading causal conclusions. The methodologies proposed herein have been successfully integrated into ByteDance's experimentation platform, demonstrating their efficacy and scalability.


% \begin{acks}

% \end{acks}

\bibliographystyle{ACM-Reference-Format}
\bibliography{KDD2}

\appendix
\onecolumn

\section{Appendix}

To facilitate the proof of our main results, we first establish several technical lemmas regarding the convergence and finite-population moments of the relevant estimators. 

\begin{lemma}[Finite-population version of the Weak Law of Large Numbers]\label{lemma1}
Assume Conditions $1-3$, the means over treatment or control of $Y_i(1), Y_i(0), X_i, Y_i(1)^2, Y_i(0)^2, X_i^2, Y_i(1)Y_i(0), Y_i(1)X_i$ and $Y_i(1)X_i$ converge in probability to the limits of the population means. 
\end{lemma}

\noindent\textbf{Proof}. This is a standard result in finite-population asymptotics. For a detailed rigorous treatment, we refer the reader to Lemma 1 in the Supplementary Material of \citet{Lin2012SupplementT}.

\begin{lemma}
Assume Conditions $1-3$. Let the limits of population covariance between potential outcomes and covariates be defined as: 
$$
C_{1, x} = \lm \frac{1}{n-1} \sm (Y_i(1) - \mu_1)(X_i - \ox), \quad C_{0, x} = \lm \frac{1}{n-1} \sm (Y_i(0) - \mu_0)(X_i - \ox). 
$$
Then, the limiting adjustment coefficients $\theta_1$ and $\theta_0$ satisfy: 
$$
\theta_1 = \frac{C_{1, x}}{S_x^2}, \quad \theta_0 = \frac{C_{0, x}}{S_x^2}. 
$$
\end{lemma}

\noindent\textbf{Proof}. We demonstrate the consistency for the treatment group and the control group follows by symmetry. The difference between the sample estimator $\wt_{1, 1}$ and the target ratio $C_{1, x} / S_x^2$ can be bounded as follows: 

$$
\left| \frac{\wc(\oy_1, \ox_1)}{\wv(\ox_1)} - \frac{C_{1, x}}{S_x^2}\right| \le \frac{\left| n\wc(\oy_1, \ox_1)\cdot S_x^2 - C_{1, x}\cdot n\wv(\ox_1)\right|}{n\wv(\ox_1)\cdot S_x^2} \le \frac{ S_x^2\cdot\left| n\wc(\oy_1, \ox_1) - C_{1, x}\right|  + C_{1, x}\cdot \left| S_x^2 - n\wv(\ox_1)\right|}{n\wv(\ox_1)\cdot S_x^2}. 
$$
By applying Lemma \ref{lemma1}, both terms in the numerator vanish in the limit: 
$$
\lm S_x^2\cdot\left| n\wc(\oy_1, \ox_1) - C_{1, x}\right| = 0, \quad \lm C_{1, x}\cdot \left| S_x^2 - n\wv(\ox_1)\right| = 0. 
$$
Consequently, we obtain the desired limit: 
$$
\theta_1 = \lm \frac{\wc(\oy_1, \ox_1)}{\wv(\ox_1)} = \frac{C_{1, x}}{S_x^2}. 
$$

\begin{lemma}[Finite population expectation and variance]\label{lem3}
Under the completely randomized design and Conditions $1-3$, the variances and covariances of the group means are given by: 
$$
\tv(\oy_1) = \frac{n - n_1}{nn_1} \cdot\frac{1}{n-1}\sm (Y_i(1) - \mu_1)^2, \quad \tv(\oy_0) = \frac{n - n_0}{nn_0} \cdot\frac{1}{n-1}\sm (Y_i(0) - \mu_0)^2, 
$$
$$
\tv(\ox_1) = \frac{n - n_1}{nn_1} \cdot\frac{1}{n-1}\sm (X_i - \ox)^2, \quad \tv(\ox_0) = \frac{n - n_0}{nn_0} \cdot\frac{1}{n-1}\sm (X_i - \ox)^2, 
$$
$$
\tc(\oy_1, \ox_1) = \frac{n - n_1}{nn_1} \cdot\frac{1}{n-1}\sm (Y_i(1) - \mu_1)(X_i - \ox), \quad \tc(\oy_0, \ox_0) = \frac{n - n_0}{nn_0} \cdot\frac{1}{n-1}\sm (Y_i(0) - \mu_0)(X_i - \ox), 
$$
$$
\tc(\oy_1, \ox_0) = -\frac{1}{n} \cdot\frac{1}{n-1}\sm (Y_i(1) - \mu_1)(X_i - \ox), \quad \tc(\oy_0, \ox_1) = -\frac{1}{n} \cdot\frac{1}{n-1}\sm (Y_i(0) - \mu_0)(X_i - \ox). 
$$
\end{lemma}

\noindent\textbf{Proof}. The randomness in our setting stems solely from the treatment assignment variable $T_i$. Recall that for a completely randomized design with $n_1$ treated units: 
$$
\text{ for } i = 1, \ldots, n , \quad  \mathbb E(T_i) = p_1, \quad \tv(T_i) = p_1p_0,  \quad \text{ for } i \neq j, \tc(T_i, T_j) = -\frac{p_1p_0}{n-1}. 
$$
Using the linearity of expectations and the properties of $T_i$, the variance of the group mean $\oy_1$ is 
$$
\begin{aligned}
\tv(\oy_1) &= \tv\lc \frac{1}{n_1} \sm T_iY_i\rc \\
& = \frac{1}{n_1^2} \lc\sm Y_i^2(1)\tv(T_i) + \sm\sum_{j\neq i}\tc(T_i, T_j)Y_i(1)Y_j(1) \rc \\
& = \frac{p_1p_0}{n_1^2(n-1)} \lc (n-1)\sm Y_i^2(1) - \sm \sum_{j\neq i}Y_i(1)Y_j(1)\rc \\
& = \frac{n - n_1}{nn_1} \cdot\frac{1}{n-1}\sm (Y_i(1) - \mu_1)^2. 
\end{aligned}
$$
The analogous calculations apply for $\tv(\oy_0)$, $\tv(\ox_1)$ and $\tv(\ox_0)$. The cross-group covariance $\tc(\oy_1, \ox_0)$ is derived similarly: 
$$
\begin{aligned}
\tc(\oy_1, \ox_0) &= \tc\lc\frac{1}{n_1} \sm T_iY_i, \frac{1}{n_0} \sm (1 - T_i)X_i\rc \\
& = \frac{1}{n_0n_1} \lc \sm\tc(T_i, 1 - T_i)Y_i(1)X_i + \sm\sum_{j\neq i}\tc(T_i, 1 - T_j)Y_i(1)X_j \rc \\
& = \frac{1}{n_0n_1} \lc -\sm p_ip_0Y_i(1)X_i + \sm\sum_{j\neq i}\frac{p_1p_0}{n-1}Y_i(1)X_j \rc \\
& = \frac{1}{n^2(n-1)} \lc -(n-1)\sm Y_i(1)X_i + \sm\sum_{j\neq i}Y_i(1)X_j \rc \\
& = -\frac{1}{n} \cdot \frac{1}{n-1}\sm (Y_i(1) - \mu_1)(X_i - \ox). 
\end{aligned}
$$
The same logic can be used to show the remaining results.

\begin{lemma}\label{lemma4}
Define that $\widetilde{\delta}_3 = \Big(\oy_1 - \theta_1(\ox_1 - \ox)\Big) - \Big(\oy_0 - \theta_0(\ox_0 - \ox)\Big)$, then we have 
$$
\lm n\lc\tv(\wdd_3) - \tv(\widetilde{\delta}_3)\rc = 0. 
$$
\end{lemma}

\noindent\textbf{Proof}. We begin by decomposing the difference between the split-sample estimator $\wdd_3$ and its oracle counterpart $\widetilde{\delta}_3$. Let 
$$
\begin{aligned}
\wdd_3 &= \Big(\oy_1 - \wt_1(\ox_1 - \ox)\Big) - \Big(\oy_0 - \wt_0(\ox_0 - \ox)\Big) \\
& = \Big(\oy_1 - \theta_1(\ox_1 - \ox)\Big) - \Big(\oy_0 - \theta_0(\ox_0 - \ox)\Big) + (\theta_1 - \wt_1)(\ox_1 - \ox) + (\wt_0 - \theta_0)(\ox_0 - \ox) \\
& =: \widetilde{\delta}_3 + \wdd_{x}. 
\end{aligned}
$$
The variance of the scaled estimator $\sqrt{n}\wdd_3$ can be expanded as 
$$
\begin{aligned}
n\tv(\wdd_3) &= \tv(\sqrt n\widetilde{\delta}_3 + \sqrt n\wdd_{x}) \\
& = n\tv(\widetilde{\delta}_3) + \tv(\sqrt n\wdd_{x}) + 2\tc(\sqrt n\widetilde{\delta}_3, \sqrt n\wdd_{x}) \\
& \le n\tv(\widetilde{\delta}_3) + \tv(\sqrt n\wdd_{x}) + 2\sqrt{n^2\tv(\widetilde{\delta}_3) \cdot \tv(\wdd_{x})}. \\
\end{aligned}
$$
To show that $\lim\limits_{n\rightarrow \infty} n\lc\tv(\wdd_3) - \tv(\widetilde{\delta}_3)\rc = 0$, it suffices to demonstrate that the remainder terms are asymptotically negligible. 
According to the central limit theorem, for $\forall \varepsilon > 0$, we have 
$$
\tv(\wdd_{x}) = \mathbb E(\wdd_{x}^2) - \mathbb E^2(\wdd_{x}) = o\lc\frac{1}{n^{2 - \varepsilon}}\rc, \quad \text{and } \tv(\widetilde{\delta}_3) = o\lc\frac{1}{n^{1 - \varepsilon}}\rc. 
$$
Therefore, we conclude 
$$
\lm n\lc\tv(\wdd_3) - \tv(\widetilde{\delta}_3)\rc = 0. 
$$
This completes the proof, confirming that the additional variability introduced by estimating $\theta$ is of a lower order than the primary sampling variance. Consequently, in the subsequent derivations, we treat $\wt$ as its probability limit $\theta$ without loss of asymptotic precision.

\subsection{Proof of Theorem \ref{relative}}

The proof of Theorem \ref{relative} proceeds in three steps: (i) establishing the asymptotic orthogonality between the adjusted group means estimators and the auxiliary adjustment terms; (ii) computing the variance shift induced by the additional covariate terms; and (iii) applying the delta method to derive the variance of the relative lift. 

\textbf{Asymptotic Orthogonality}. Using Lemma \ref{lem3}, we first evaluate the covariance between the adjusted estimator $\ty_1^*$ ($\ty_0^*$) and the additional adjustment term $\wt_{3, 0}(\ox_0 - \ox)$. As $n \to \infty$, we have: 
$$
\begin{aligned}
\lm n\tc\lc\ty_1^*, \wt_{3, 0}(\ox_0 - \ox)\rc &= \lm n\tc\lc\oy_1 - \theta_1(\ox_1 - \ox), \theta_0(\ox_0 - \ox)\rc \\
&= \lm \lc-\theta_0\cdot\frac{1}{n-1}\sm(Y_i(1) - \mu_1)(X_i - \ox) + \theta_1\theta_0\cdot\frac{1}{n-1}\sm(X_i - \ox)^2\rc \\
&= -\theta_0C_{1, x} + \theta_1\theta_0 S_x^2 \\
&= 0. 
\end{aligned}
$$
Analogously, 
$$
\begin{aligned}
\lm n\tc\lc\ty_0^*, \wt_{3, 0}(\ox_0 - \ox)\rc &= \lm n\tc\lc\oy_0 - \theta_0(\ox_0 - \ox), \theta_0(\ox_0 - \ox)\rc \\
&= \lm \lc\theta_0\frac{p_1}{p_0}\cdot\frac{1}{n-1}\sm(Y_i(0) - \mu_0)(X_i - \ox) - \theta_0^2\frac{p_1}{p_0}\cdot\frac{1}{n-1}\sm(X_i - \ox)^2\rc \\
&= \theta_0^2\frac{p_1}{p_0}S_x^2 - \theta_0^2\frac{p_1}{p_0} S_x^2 \\
& = 0. 
\end{aligned}
$$
We further calculate the asymptotic variance of the covariate term: 
$$
\begin{aligned}
\lm n\tv\lc\wt_{3, 0}(\ox_0 - \ox)\rc &= \lm n\tv\lc\theta_0(\ox_0 - \ox)\rc = \frac{p_1}{p_0}\theta_0^2S_x^2. 
\end{aligned}
$$

\textbf{Variance Decomposition}. Based on the orthogonality shown above, the variance increments for the treatment and control estimators are 

$$
\lm n \lc \tv(\ty_1^{**}) - \tv(\ty_1^{*}) \rc = \lm n \lc \tv(\ty_1^{*} + \wt_{3, 0}(\ox_0 - \ox)) - \tv(\ty_1^{*}) \rc = \frac{p_1}{p_0}\theta_0^2S_x^2, 
$$
$$
\lm n \lc \tv(\ty_0^{**}) - \tv(\ty_0^{*}) \rc = \lm n \lc \tv(\ty_0^{*} - \wt_{3, 0}(\ox_0 - \ox)) - \tv(\ty_0^{*}) \rc = \frac{p_1}{p_0}\theta_0^2S_x^2. 
$$

\textbf{Delta Method for Relative Lift}. Applying the Taylor expansion for the variance of a ratio (the delta method): 
$$
\begin{aligned}
\tv\lc\frac{\ty^{**}_1}{\ty^{**}_0}\rc - \tv\lc\frac{\ty^{*}_1}{\ty^{*}_0}\rc &= \frac{\mathbb E^2(\ty_1^{**})}{\mathbb E^2(\ty_0^{**})}\lc\frac{\tv(\ty_1^{**})}{\mathbb E^2(\ty_1^{**})} + \frac{\tv(\ty_0^{**})}{\mathbb E^2(\ty_0^{**})} - \frac{2\tc(\ty_1^{**}, \ty_0^{**})}{\mathbb E(\ty_1^{**})\mathbb E(\ty_0^{**})} \rc - \frac{\mathbb E^2(\ty_1^{*})}{\mathbb E^2(\ty_0^{*})}\lc\frac{\tv(\ty_1^{*})}{\mathbb E^2(\ty_1^{*})} + \frac{\tv(\ty_0^{*})}{\mathbb E^2(\ty_0^{*})} - \frac{2\tc(\ty_1^{*}, \ty_0^{*})}{\mathbb E(\ty_1^{*})\mathbb E(\ty_0^{*})} \rc + o\lc\frac{1}{n}\rc. 
\end{aligned}
$$
Substituting the variance increments from variance decomposition into the above formula, we obtain the final result: 
$$
\begin{aligned}
\lm n\lc\tv\lc\frac{\ty^{**}_1}{\ty^{**}_0}\rc - \tv\lc\frac{\ty^{*}_1}{\ty^{*}_0}\rc \rc&= \frac{\lb^2}{\la^2}\lc \frac{p_0^{-1}p_1\theta_0^2S_x^2}{\lb^2} + \frac{p_0^{-1}p_1\theta_0^2S_x^2}{\la^2} - \frac{2p_0^{-1}p_1\theta_0^2S_x^2}{\la\lb} \rc \\
& = \frac{\lb^2}{\la^2} \lc \frac{1}{\lb} - \frac{1}{\la} \rc^2 \frac{p_1}{p_0}\theta_0^2S_x^2. 
\end{aligned}
$$
This completes the proof. 

\subsection{Proof of Theorem \ref{61}}

First, we establish the asymptotic variance for the ATE estimator $\widehat{\tau}_1$ which uses the full covariate sample mean. Define that 
$$
\omega_1 = \frac{1}{m}\sum_{i = 1}^m W_i(1), \quad \omega_0 =  \frac{1}{m}\sum_{i = 1}^m W_i(0). 
$$
Then, 
$$
\begin{aligned}
\lnm m\tv(\wa_1) & = \lnm m\tv\lc \ow_1 - \ow_0 - \wl_1(\oz_1 - \oz) + \wl_0(\oz_0 - \oz)\rc \\
& = \lnm m\tv\lc \ow_1 - \ow_0 - \lambda_1(\oz_1 - \oz) + \lambda_0(\oz_0 - \oz)\rc \\
& = \lnm m\lc\tv( \ow_1 - \lambda_1\oz_1) + \tv( \ow_0 - \lambda_0\oz_0) - 2\tc(\ow_1 - \lambda_1\oz_1, \ow_0 - \lambda_0\oz_0)\rc \\
& = \frac{q_0}{q_1} S_{1, z}^2 + \frac{q_1}{q_0}S_{0, z}^2 + \lnm \frac{2}{m -1}\smm \lc W_i(1) - \omega_1 - \lambda_1(Z_i - \oz)\rc\lc W_i(0) - \omega_0 - \lambda_0(Z_i - \oz)\rc\\
& = \frac{q_0}{q_1} S_{1, z}^2 + \frac{q_1}{q_0}S_{0, z}^2 + S_{1, z}^2 + S_{0, z}^2 - S_{\tau, z}^2 \\
& = \frac{1}{q_1} S_{1, z}^2 + \frac{1}{q_0} S_{0, z}^2 - S_{\tau, z}^2. 
\end{aligned}
$$

To analyze $\widehat{\tau}_2$, we decompose it into the sum of $\widehat{\tau}_1$ and a correction term involving the two-group covariate means $\overline{Z}_{0, 1}$. A key step is showing that the adjusted estimator is asymptotically orthogonal to this correction term. Specifically, consider the asymptotic covariance: 
$$
\begin{aligned}
&\quad  \lnm m\tc\lc \ow_1 - \ow_0 - \lambda_1\oz_1 + \lambda_0\oz_0, (\lambda_1 - \lambda_0)\oz_{0, 1} \rc \\
& = \lnm m\tc\lc \ow_1 - \ow_0 - \lambda_1\oz_1 + \lambda_0\oz_0, (\lambda_1 - \lambda_0)\lc\frac{q_1}{1 - q_2}\oz_1 + \frac{q_0}{1 - q_2}\oz_0\rc \rc \\
& = \lnm m(\lambda_1 - \lambda_0)\frac{q_1}{1 - q_2}\lc\tc\lc \ow_1 - \ow_0 - \lambda_1\oz_1 + \lambda_0\oz_0, \oz_1\rc \rc + \lnm m(\lambda_1 - \lambda_0)\frac{q_0}{1 - q_2}\lc\tc\lc \ow_1 - \ow_0 - \lambda_1\oz_1 + \lambda_0\oz_0, \oz_0\rc \rc \\
& = 0. 
\end{aligned}
$$

Finally, using the orthogonality, the variance of $\widehat{\tau}_2$ can be expressed as the sum of the variance of $\widehat{\tau}_1$ and the variance of the auxiliary term: 
$$
\begin{aligned}
\lnm m\tv(\wa_2) & = \lnm m\tv\lc \ow_1 - \ow_0 - \wl_1(\oz_1 - \oz_{0, 1}) + \wl_0(\oz_0 - \oz_{0, 1})\rc \\
&= \lnm m\tv\lc \ow_1 - \ow_0 - \lambda_1\oz_1 + \lambda_0\oz_0 + (\lambda_1 - \lambda_0)\oz_{0, 1}\rc \\
& = \lnm m\lc\tv(\ow_1 - \ow_0 - \lambda_1\oz_1 + \lambda_0\oz_0) + \tv\lc(\lambda_1 - \lambda_0)\oz_{0, 1}\rc + 2\tc\lc \ow_1 - \ow_0 - \lambda_1\oz_1 + \lambda_0\oz_0, (\lambda_1 - \lambda_0)\oz_{0, 1}\rc \rc\\
& = \frac{1}{q_1} S_{1, z}^2 + \frac{1}{q_0} S_{0, z}^2 - S_{\tau, z}^2 + \lnm m(\lambda_1 - \lambda_0)^2\tv\lc\frac{q_0}{1 - q_2}\oz_0 + \frac{q_1}{1 - q_2}\oz_1\rc \\
& = \frac{1}{q_1} S_{1, z}^2 + \frac{1}{q_0} S_{0, z}^2 -  S_{\tau, z}^2 + \frac{q_2}{1 - q_2}(\lambda_1 - \lambda_0)^2S_z^2. 
\end{aligned}
$$
This completes the proof. 

\subsection{Proof of Theorem \ref{62}}

According to the properties of sample moments established in Lemma \ref{lem3}, the standard variance estimator for $\widehat{\tau}_1$ converges to: 
$$
\begin{aligned}
\lnm m\wv(\widehat{\tau}_1) &= \lnm \lc\frac{1}{q_1(m_1-1)}\sum_{i=1}^m\mathbb I(A_i=1)\big( W_i - \ow_1 - \widehat{\lambda}_1(Z_i - \oz_1)\big)^2 + \frac{1}{q_0(m_0-1)}\sum_{i=1}^m\mathbb I(A_i=0)\big( W_i - \ow_0 - \widehat{\lambda}_0(Z_i - \oz_0)\big)^2\rc \\
& = \frac{1}{q_1}S_{1, z}^2 + \frac{1}{q_0}S_{0, z}^2. 
\end{aligned}
$$

Comparing this with the true variance derived in Theorem \ref{61}, we can quantify the conservativeness: 
$$
\lnm m\lc\wv(\widehat{\tau}_1) - \tv(\widehat{\tau}_1)\rc = S_{\tau, z}^2. 
$$

Similarly, for the ATE estimator $\widehat{\tau}_2$, the difference is: 
$$
\lnm m\lc\wv(\widehat{\tau}_2) - \tv(\widehat{\tau}_2)\rc = S_{\tau, z}^2 - \frac{q_2}{1 - q_2}(\lambda_1 - \lambda_0)^2S_z^2. 
$$

\subsection{Proof of Theorem \ref{71}}

Let $P_i = 1$ indicate that individual $i$ is sampled and assigned to treatment, and $P_i = 0$ otherwise. Similarly, let $Q_i = 1$ indicate that individual $i$ is sampled and assigned to control, and $Q_i = 0$ otherwise. We first characterize the probability limits of the various $\theta$ estimators. By applying Lemma \ref{lemma1}, we obtain: 
$$
\lm\wt_{1, 1} = \lm\frac{\wc(\oy, \ox)}{\wv(\ox)} = \lm\frac{1}{n\wv(\ox)}\cdot n\lc \frac{n_1(n_1 - 1)}{n(n - 1)}\wc(\oy_1, \ox_1) + \frac{n_0(n_0 - 1)}{n(n - 1)}\wc(\oy_0, \ox_0) + o\Big( \frac{1}{n} \Big)\rc = \frac{p_1C_{1, x} + p_0C_{0, x}}{S_x^2} = p_1\theta_1 + p_0\theta_0 = \theta_{\text{Full}}. 
$$
Similarly, for the pooled adjustment 
$$
\lm\wt_{2, 1} = \lm\frac{\wc(\oy_1, \ox_1) + \wc(\oy_0, \ox_0)}{\wv(\ox_1) + \wv(\ox_0)} = \lm\frac{n\wc(\oy_1, \ox_1) + n\wc(\oy_0, \ox_0)}{n\wv(\ox_1) + n\wv(\ox_0)} = \frac{p_1^{-1}C_{1, x} + p_0^{-1}C_{0, x}}{p_1^{-1}S_x^2 + p_0^{-1}S_x^2} = p_0\theta_1 + p_1\theta_0 = \theta_{\text{Pooled}}. 
$$

Within the framework of two-stage sampling, we still present the properties such as the variance and covariance of the assignment variables $P_i$ and $Q_i$: 
$$
\text{ For } i=1, \ldots, N, \quad \mathbb E(P_i) = pp_1, \quad \mathbb E(Q_i) = pp_0, \quad\tv(P_i) = pp_1(1 - pp_1),\quad \tv(Q_i) = pp_0(1 - pp_0), 
$$
$$
\tc(P_i, Q_i) = -p^2p_1p_0, \quad \text{ for } i\neq j, \tc(P_i, Q_j) = \frac{n_1n_0}{N(N-1)} - \frac{n_1n_0}{N^2} = \frac{p^2p_1p_0}{N-1}. 
$$
Using these assignment properties, we derive the variance of the mean differences for the outcome and covariates: 
$$
\begin{aligned}
&\quad \tv(\oy_1 - \oy_0)  \\
&= \tv\lc \frac{1}{n_1}\sn P_iY_i - \frac{1}{n_0}\sn Q_i Y_i \rc \\
& = \tv\lc \frac{1}{n_1}\sn P_iY_i \rc + \tv\lc \frac{1}{n_0}\sn Q_iY_i \rc - 2\tc\lc \frac{1}{n_1}\sn P_iY_i,  \frac{1}{n_0}\sn Q_iY_i\rc \\
& = \frac{N - n1}{Nn_1} \cdot \frac{1}{N - 1}\sn (Y_i(1) - \mu_1)^2 + \frac{N - n0}{Nn_0}  \cdot \frac{1}{N - 1}\sn (Y_i(0) - \mu_0)^2 + \frac{2}{N} \cdot \frac{1}{N - 1}\sn (Y_i(1) - \mu_1)(Y_i(0) - \mu_0) \\
& = \frac{N - n1}{Nn_1}  \cdot \frac{1}{N - 1}\sn (Y_i(1) - \mu_1)^2 + \frac{N - n0}{Nn_0}  \cdot \frac{1}{N - 1}\sn (Y_i(0) - \mu_0)^2 \\
& \quad + \frac{1}{N} \cdot \frac{1}{N - 1}\lc\sn (Y_i(1) - \mu_1)^2 + \sn (Y_i(0) - \mu_0)^2 - \sn \lc(Y_i(1) - \mu_1) - (Y_i(0) - \mu_0)\rc^2\rc \\
& =  \frac{1}{n_1}  \cdot \frac{1}{N - 1}\sn (Y_i(1) - \mu_1)^2 + \frac{1}{n_0}  \cdot \frac{1}{N - 1}\sn (Y_i(0) - \mu_0)^2 - \frac{1}{N}  \cdot \frac{1}{N - 1}\sn \lc(Y_i(1) - \mu_1) - (Y_i(0) - \mu_0)\rc^2, 
\end{aligned}
$$
and 
$$
\begin{aligned}
\tv(\ox_1 - \ox_0)  &= \tv\lc \frac{1}{n_1}\sn P_iX_i - \frac{1}{n_0}\sn Q_i X_i \rc \\
& = \tv\lc \frac{1}{n_1}\sn P_iX_i \rc + \tv\lc \frac{1}{n_0}\sn Q_iX_i \rc - 2\tc\lc \frac{1}{n_1}\sn P_iX_i,  \frac{1}{n_0}\sn Q_iX_i\rc \\
& = \frac{N - n1}{Nn_1}  \cdot \frac{1}{N - 1}\sn (X_i - \ox)^2 + \frac{N - n0}{Nn_0}  \cdot \frac{1}{N - 1}\sn (X_i - \ox)^2 + \frac{2}{N} \cdot \frac{1}{N - 1} \sn (X_i - \ox)^2 \\
& =  \frac{n}{n_1n_0}  \cdot \frac{1}{N - 1} \sn (X_i - \ox)^2. 
\end{aligned}
$$
Substituting the above limits, we obtain the asymptotic variances for $\widehat{\delta}_1, \widehat{\delta}_2$, and $\widehat{\delta}_3$. For the full-sample estimation: 
$$
\begin{aligned}
\lm n\tv(\wdd_1) & = \lm n\tv\lc\oy_1 - \oy_0 - \wt_{1, 1}(\ox_1 - \ox_0)\rc \\
& = \lm n\tv\lc\oy_1 - \oy_0 - \theta_{\text{Full}}(\ox_1 - \ox_0)\rc \\
& = \frac{1}{p_1}S_1^2 + \frac{1}{p_0}S_0^2 - pS_{\delta}^2 + \frac{1}{p_1p_0}\theta_{\text{Full}}^2S_x^2 - 2\theta_{\text{Full}} \lc \frac{\theta_1}{p_1}+\frac{\theta_0}{p_0} \rc S_x^2 \\
& = \frac{1}{p_1}S_1^2 + \frac{1}{p_0}S_0^2 - pS_{\delta}^2 + \frac{1}{p_1p_0}\lc\theta_{\text{Full}}^2 - 2\theta_{\text{Full}}\theta_{\text{Pooled}}\rc S_x^2. 
\end{aligned}
$$
For the pooled-sample estimator $\widehat{\delta}_2$, the variance simplifies to: 
$$
\begin{aligned}
\lm n\tv(\wdd_2) & = \lm n\tv\lc\oy_1 - \oy_0 - \wt_{2, 1}(\ox_1 - \ox_0)\rc \\
& = \lm n\tv\lc\oy_1 - \oy_0 - \theta_{\text{Pooled}}(\ox_1 - \ox_0)\rc \\
& = \frac{1}{p_1}S_1^2 + \frac{1}{p_0}S_0^2 - pS_{\delta}^2 + \frac{1}{p_1p_0}\theta_{\text{Pooled}}^2S_x^2 - 2\theta_{\text{Pooled}} \lc \frac{\theta_1}{p_1}+\frac{\theta_0}{p_0} \rc S_x^2 \\
& = \frac{1}{p_1}S_1^2 + \frac{1}{p_0}S_0^2 - pS_{\delta}^2 - \frac{1}{p_1p_0}\theta_{\text{Pooled}}^2 S_x^2. 
\end{aligned}
$$
For the split-sample estimator $\widehat{\delta}_3$, the variance simplifies to: 
$$
\begin{aligned}
\lm n\tv(\wdd_3) & = \lm n\tv\lc\oy_1 - \oy_0 - \wt_{3, 1}(\ox_1 - \ox) + \wt_{3, 0}(\ox_0 - \ox)\rc \\
& = \lm n\tv\lc\oy_1 - \oy_0 - p_0\wt_{3, 1}(\ox_1 - \ox_0) + p_1\wt_{3, 0}(\ox_0 - \ox_1)\rc \\
& = \lm n\tv\lc\oy_1 - \oy_0 - (p_0\wt_{3, 1} + p_1\wt_{3, 0})(\ox_1 - \ox_0) \rc \\
& = \lm n\tv\lc\oy_1 - \oy_0 - \theta_{\text{Pooled}}(\ox_1 - \ox_0)\rc \\
& = \frac{1}{p_1}S_1^2 + \frac{1}{p_0}S_0^2 - pS_{\delta}^2 - \frac{1}{p_1p_0}\theta_{\text{Pooled}}^2 S_x^2. 
\end{aligned}
$$
By comparing these expressions, the efficiency gains of $\widehat{\delta}_2$ and $\widehat{\delta}_3$ over $\widehat{\delta}_1$ are: 
$$
\lm n\lc \tv(\wdd_1) - \tv(\wdd_2)\rc  = \lm n\lc \tv(\wdd_1) - \tv(\wdd_3)\rc  = \frac{1}{p_1p_0}\lc \theta_{\text{Full}} - \theta_{\text{Pooled}} \rc^2S_x^2. 
$$
This proves that the pooled-sample and split-sample estimator are asymptotically at least as efficient as the full-sample adjustment.

\subsection{Proof of Theorem \ref{72}}

Applying the consistency of sample moments, the standard variance estimators for the three methods converge to: 
$$
\begin{aligned}
\lm n\wv(\wdd_1) &= \lm\lc \frac{1}{p_1} \sn P_i(Y_i - \oy_1 - \wt_{1, 1}(\ox_i - \ox_1))^2 + \frac{1}{p_0} \sn Q_i(Y_i - \oy_0 - \wt_{1, 0}(\ox_i - \ox_0))^2\rc \\
& = \frac{1}{p_1}\lc S_1^2 + \theta_{\text{Full}}^2S_x^2 - 2\theta_{\text{Full}} C_{1, x} \rc + \frac{1}{p_0}\lc S_0^2 + \theta_{\text{Full}}^2S_x^2 - 2\theta_{\text{Full}} C_{0, x} \rc \\
& = \frac{1}{p_1} S_1^2 + \frac{1}{p_0} S_0^2 + \frac{1}{p_1p_0}\lc\theta_{\text{Full}}^2 - 2\theta_{\text{Full}}\theta_{\text{Pooled}}\rc S_x^2. 
\end{aligned}
$$

$$
\begin{aligned}
\lm n\wv(\wdd_2) &= \lm\lc \frac{1}{p_1} \sn P_i(Y_i - \oy_1 - \wt_{2, 1}(\ox_i - \ox_1))^2 + \frac{1}{p_0} \sn Q_i(Y_i - \oy_0 - \wt_{2, 0}(\ox_i - \ox_0))^2\rc \\
& = \frac{1}{p_1}\lc S_1^2 + \theta_{\text{Pooled}}^2S_x^2 - 2\theta_{\text{Pooled}} C_{1, x} \rc + \frac{1}{p_0}\lc S_0^2 + \theta_{\text{Pooled}}^2S_x^2 - 2\theta_{\text{Pooled}} C_{0, x} \rc \\
& = \frac{1}{p_1} S_1^2 + \frac{1}{p_0} S_0^2 - \frac{1}{p_1p_0}\theta_{\text{Pooled}}^2 S_x^2. 
\end{aligned}
$$

$$
\begin{aligned}
\lm n\wv(\wdd_3) &= \lm\lc \frac{1}{p_1} \sn P_i(Y_i - \oy_1 - \wt_{3, 1}(\ox_i - \ox_1))^2 + \frac{1}{p_0} \sn Q_i(Y_i - \oy_0 - \wt_{3, 0}(\ox_i - \ox_0))^2\rc \\
& = \frac{1}{p_1}\lc S_1^2 + \theta_1^2S_x^2 - 2\theta_1 C_{1, x} \rc + \frac{1}{p_0}\lc S_0^2 + \theta_0^2S_x^2 - 2\theta_0 C_{0, x} \rc \\
& = \frac{1}{p_1} S_1^2 + \frac{1}{p_0} S_0^2 - \lc \frac{1}{p_1}\theta_1^2 + \frac{1}{p_0}\theta_0^2 \rc S_x^2. 
\end{aligned}
$$
Comparing these with the true asymptotic variances from Theorem \ref{71}, we find the following estimation gaps: 
$$
\lm n\lc\wv(\wdd_1) - \tv(\wdd_1)\rc= pS_{\delta}^2, 
$$
$$
\lm n\lc\wv(\wdd_2) - \tv(\wdd_2)\rc= pS_{\delta}^2, 
$$
$$
\lm n\lc\wv(\wdd_3) - \tv(\wdd_3)\rc = pS_{\delta}^2 - (\theta_1 - \theta_0)^2S_x^2. 
$$
This completes the proof.

\end{document}